\documentclass{article}
\usepackage{graphicx} % Required for inserting images
\usepackage{amsmath}
\usepackage{amsthm}
\usepackage{amsfonts} 

\newtheorem{theorem}{Theorem}      % Theorem 1, Theorem 2, ...
\newtheorem{lemma}{Lemma}
\newtheorem{proposition}{Proposition}

\usepackage{natbib}

\title{No Last Mile: A Theory of the Human Data Market}
\author{Ali Ansari, Mark Esposito, Ava Fitoussy, Liu Zhang}
\date{February 2026}

\begin{document}

\maketitle

\begin{abstract}
The standard framing treats structured human-data work as transitional, a bridge between today’s imperfect models and a future state where automation is complete. We challenge this view by modeling \emph{structured human data} as a persistent production input: evaluation, rubric-based judgment, auditing, exception handling, and continual updates that convert raw model capability into dependable, deployable performance. These activities accumulate into a reusable AI capability stock that raises productivity by improving reliability on existing tasks and by expanding the frontier of task families for which AI can be used at high confidence. Crucially, this capability stock depreciates as tasks and contexts drift, standards evolve, and new edge cases emerge. In a tractable baseline model, an interior steady state implies a closed-form, strictly positive long-run labor share devoted to structured human-data work whenever depreciation is positive, a ``no last mile'' result in which maintenance demand persists even as models improve. We then microfound aggregate capability with a portfolio of task families featuring diminishing returns, frontier entry, and complementarity, generating reallocation toward low-maturity and bottleneck families and a Roy-style mechanism for within-structured wage dispersion. Finally, we map model objects to observable proxies using standard data layers, and provide a conservative calibration suggesting a 5-7$\%$ steady-state structured labor share in the long run.
\end{abstract}

\section{Introduction}
Generative AI has moved from novelty to mainstream workplace technology in a remarkably short time. Survey evidence shows rapid and broad take-up: by late 2024, nearly 40$\%$ of U.S. adults aged 18–64 reported using generative AI, and about 23$\%$ of employed respondents reported using it for work at least once in the prior week \citep{bick2026rapid}. At the same time, firm-level rollout studies show that generative AI can raise productivity in real operational settings. In a large customer-support setting, access to an AI assistant increased issues resolved per hour by about 15$\%$ on average, with much larger gains for less experienced workers \citep{brynjolfsson2025generative}. These findings suggest broad scope for AI to change how work is done and how output is produced, but they also highlight a key measurement challenge: productivity gains observed on well-defined, AI-suitable tasks need not translate one-for-one into aggregate gains. Benefits are uneven across tasks, depend on correct use and workflow redesign, and often require complementary investments, including monitoring, exception handling, updating procedures, and building organizational know-how, before they scale beyond pilot settings. In task-based macroeconomic frameworks, aggregate outcomes therefore hinge on which parts of the task distribution are effectively automated versus complemented, and how quickly the economy expands the set of tasks for which AI is reliably usable \citep{acemoglu2025simple}. This emphasis on task reallocation also aligns with canonical “tasks, not jobs” perspectives in which technology displaces some tasks while creating or reshaping others, so equilibrium effects depend on the pace of task creation and reorganization \citep{autor2015job,acemoglu2018race}.

Early labor-market evidence is consistent with this “reorganization rather than displacement” view. Using representative adoption surveys linked to administrative labor records in Denmark, \citet{humlum2025large} find precisely estimated near-zero effects on earnings and recorded hours over the first two years, ruling out changes larger than about 2$\%$. This pattern is consistent with the idea that the near-term adjustment margin is often reorganization of tasks, not headline job loss: integrating AI into workflows, managing quality, dealing with exceptions, and updating practices as conditions evolve. This motivates a closer look at an ingredient that is central in practice but often under-modeled in macro and technology-adoption frameworks: the ongoing production of structured human data that converts raw model capability into dependable performance.

Our paper formalizes this idea by modeling a structured human-data sector, activities such as evaluation, rubric-based judgment, auditing, exception handling, and continual updates, as a distinct input into production that accumulates a reusable AI capability stock. This stock raises productivity in the economy not only by improving performance on existing AI-exposed tasks, but also by enabling frontier expansion: it increases the set of task families for which AI can be deployed with sufficiently high reliability, thereby shifting the effective “coverage” of AI across the task distribution over time. Because structured outputs are reusable and can be combined across users and firms, they also resemble nonrival data inputs that can generate increasing returns—while still requiring governance and ongoing production to remain valuable as contexts evolve \citep{jones2020nonrivalry}. Crucially, capability stock also depreciates as tasks and contexts drift, standards change, and new edge cases emerge.

Capability depreciation ($\delta_k$) in our model should not be interpreted as transient “statistical decay” remediable by scale alone. In deployed systems, the mapping between codified representations and the world changes over time, via dataset shift and concept drift, so that even accurate codification can lose alignment as contexts evolve \citep{quinonero2008dataset, gama2014survey}. Production-oriented ML frameworks similarly emphasize recurring operational costs, including monitoring, debugging, and updating models and data pipelines, as environments and dependencies change, rather than one-off fixes \citep{sculley2015hidden}. Therefore, we interpret structured data work not only as accumulating an “expertise stock,” but as maintaining relational coherence between AI systems and evolving external structure, including workflows, standards, and user needs \citep{gassman2025coherence}. Historically, successive waves of codification and standardization shifted where expertise is applied but did not eliminate the need for expert maintenance, since formal standards are themselves periodically reviewed and revised \citep{nelson1985evolutionary, rosenberg1982inside}.

Empirically, $\delta_k$ bundles at least three distinct drift channels. \emph{Technological drift} arises when model, retrieval, tooling, or product-surface updates break previously reliable rubrics and evaluation harnesses. \emph{Environmental drift} reflects changes in the underlying task distribution (new intents, new content, seasonality, adversarial behavior) that create novel edge cases even with a fixed model. \emph{Organizational drift} captures shifting standards and stakeholder objectives (policy updates, brand constraints, risk tolerance, or legal requirements) that redefine what counts as a “good” output. Distinguishing these channels provides measurement guidance: technological drift can be proxied by discrete “breakage” around version changes; environmental drift by rising error or disagreement conditional on stable rubrics; and organizational drift by guideline churn and relabeling episodes. A practical empirical approach is to treat $\delta_k$ as an event-driven hazard rate of maturity loss and to estimate $\delta_k=\delta^{tech}+\delta^{env}+\delta^{org}$ using event-study variation around each type of shock.

In equilibrium, this generates a “no last mile” implication: even as foundation models improve, there is no point at which demand for structured human-data work disappears, because maintaining alignment with a changing world requires continuous investment. This perspective complements task-based macro models \citep{acemoglu2025simple} and helps reconcile rapid adoption and micro-level gains \citep{bick2026rapid,brynjolfsson2025generative} with muted short-run labor-market effects \citep{humlum2025large}.

The framework also clarifies why compensation in structured human-data work can be highly dispersed across activities. Structured work spans a wide range of task families, from routinized labeling and standardized checks to expert judgment, domain-specific auditing, and the design of tightly specified rubrics and prompts. In equilibrium, workers sort across these families based on comparative advantage: complex, high-stakes families that reward calibrated judgment and domain knowledge draw from a thinner supply of qualified labor, often requiring screening, training, or accumulated experience, and therefore command a wage premium relative to more easily substitutable work. By contrast, simpler structured tasks with low skill requirements need not exhibit such premia and may be priced near the global opportunity cost of general labor. 

A key implication is dynamic: because structured effort is “leveraged” through the capability stock, improving reliability and expanding the set of tasks that can be deployed, frontier growth reallocates structured work toward the families where drift and exception risk remain high, rather than eliminating demand altogether. This provides a natural counterpoint to the view that structured human-data workers are merely “training replacements”: even as automation advances, the long-run equilibrium features persistent, and potentially rising, demand for high-judgment structured work that maintains alignment as tasks, standards, and environments evolve.
Section \ref{sec:model3} formalizes this dispersion with a Roy-style assignment across task families, where family-specific shadow prices and comparative advantage sorting generate wage premia in thin-supply, high-judgment families and commodity wages in standardized families \citep{roy1951earnings}.

The rest of the paper is structured as follows. Section 2 presents the baseline model and characterizes the equilibrium that delivers the no last mile prediction. Section 3 extends the framework to task families, deriving reallocation across families, frontier expansion, and a transparent mechanism for within-structured wage dispersion through comparative-advantage sorting. Section 4 maps the key objects in the model to measurable proxies, reports a conservative calibration to discipline magnitudes, and outlines an empirical roadmap for testing the paper’s core propositions using standard data sources in the AI labor-market literature. Section 5 concludes.

% =========================
% Section 2: Baseline Model
% =========================
\section{Baseline Model}
\label{sec:modelA}

This section presents a tractable baseline model that delivers a closed-form steady-state share of labor allocated to structured human-data work. The key feature is that structured work accumulates an intangible capability stock that raises aggregate productivity.

\subsection{Environment and sector definitions}

Time is discrete, indexed by $t=0,1,2,\dots$. The economy has a fixed labor endowment $\bar L>0$ each period, allocated between:
\begin{itemize}
  \item \textbf{Unstructured production sector ($U$).} Labor $L_{U,t}$ produces final services/goods delivered case-by-case.
  \item \textbf{Structured human-data sector ($S$).} Labor $L_{S,t}$ produces AI training/evaluation data artifacts that accumulate into a reusable capability stock.
\end{itemize}

Labor feasibility is
\begin{equation}
L_{U,t}+L_{S,t}=\bar L,\qquad L_{U,t},L_{S,t}\ge 0.
\label{eq:labor_feas}
\end{equation}

There is a predetermined capital stock $K>0$ (fixed in the baseline). Aggregate output is
\begin{equation}
Y_t = A(k_t)\, K^\alpha\, L_{U,t}^{\,1-\alpha},\qquad \alpha\in(0,1),
\label{eq:production}
\end{equation}
where $k_t> 0$ is stock of “codified knowledge / AI capability” that raises economy-wide productivity. We assume
\begin{equation}
A(k) = \bar A\,k^{\gamma},\qquad \bar A>0,\ 1>\gamma>0.
\label{eq:A_of_k}
\end{equation}

\subsection{Capability accumulation}

Structured labor accumulates $k_t$ according to
\begin{equation}
k_{t+1} = (1-\delta_k)\,k_t + \phi(L_{S,t}),
\qquad \delta_k\in(0,1),\ \eta>0.
\label{eq:k_law}
\end{equation}
The depreciation term $\delta_k$ captures obsolescence or drift due to distribution shifts, changing objectives, new contexts, and related frontier expansion forces. We assume $\phi(L_{S,t})=\eta L_{S,t}$, where $\eta$ denotes how effectively structured labor converts into capability, assumed to be constant across time on the aggregate level for simplicity. 

\subsection{Competitive wages and labor allocation}

Let $w_{U,t}$ denote the wage in unstructured sector (marginal product of labor). From Equation (\ref{eq:production}),
\begin{equation}
w_{U,t} \equiv \frac{\partial Y_t}{\partial L_{U,t}}
= (1-\alpha)\,A(k_t)\,K^\alpha\,L_{U,t}^{-\alpha}.
\label{eq:wU}
\end{equation}

Interpret “structured sector” as competitive contracting (or a platform passing through the marginal value). Then the wage equals the marginal value product of structured labor.

Structured labor raises $k_t$ by $\eta$. Let $v_t$ denote the shadow value of one additional unit of capability $k_t$, then 
\begin{equation}
w_{S,t} = \eta\,v_t
\label{eq:wS}
\end{equation}

In steady state, the value of an extra unit of $k$ is the present value of its marginal contribution to output:
\begin{equation}
v = \frac{\partial Y/\partial k}{r+\delta_k},
\qquad r>0,
\label{eq:shadow_value}
\end{equation}
where $r$ is a discount rate and $\delta_k$ accounts for effective depreciation.

Given Equation(\ref{eq:A_of_k}) and Equation (\ref{eq:production}),
\begin{equation}
\frac{\partial Y_t}{\partial k_t}
= A'(k_t)\,K^\alpha\,L_{U,t}^{1-\alpha}
= \gamma\,\bar A k_t^{\gamma-1}\,K^\alpha\,L_{U,t}^{1-\alpha}.
\label{eq:dYdk}
\end{equation}

\paragraph{Labor allocation condition.}
In the baseline, labor is freely mobile between sectors and workers take wages as given. An interior allocation satisfies
\begin{equation}
w_{S,t} = w_{U,t}.
\label{eq:w_equal}
\end{equation}

\subsection{Steady state}

A steady state is a triple $(k^*,L_S^*,L_U^*)$ such that:
(i) $k_{t+1}=k_t=k^*$, (ii) $L_U^*=\bar L-L_S^*$, and (iii) Equation (\ref{eq:w_equal}) holds (interior case).

From Equation (\ref{eq:k_law}), the steady-state capability stock satisfies
\begin{equation}
k^* = \frac{\eta}{\delta_k}\,L_S^*.
\label{eq:k_star}
\end{equation}

\begin{theorem}[Existence and uniqueness of an interior steady state]
\label{thm:exist_unique_A}
Suppose $\alpha\in(0,1)$, $\gamma>0$, $r>0$, $\delta_k\in(0,1)$, $\eta>0$, and $\kappa>0$. Then there exists a unique steady state with $0<L_S^*<\bar L$ and $k^*>0$. Moreover, the steady-state structured labor share
$s^*\equiv L_S^*/\bar L$ is given by
\begin{equation}
s^*
= \frac{\gamma\,\delta_k}{\gamma\,\delta_k + (1-\alpha)\,(r+\delta_k)}.
\label{eq:s_star}
\end{equation}
\end{theorem}

\begin{proof}
At steady state, Equation (\ref{eq:k_star}) holds. Using Equation (\ref{eq:wU}), (\ref{eq:wS}), and (\ref{eq:dYdk}), the interior condition Equation (\ref{eq:w_equal}) becomes
\[
\eta\cdot \frac{\gamma\,\bar A\,(k^*)^{\gamma-1}\,K^\alpha\,(L_U^*)^{1-\alpha}}{r+\delta_k}
=
(1-\alpha)\,\bar A\,(k^*)^\gamma\,K^\alpha\,(L_U^*)^{-\alpha}.
\]
Cancel common factors and simplify to obtain
\[
\eta \cdot \frac{\gamma}{r+\delta_k}\cdot \frac{L_U^*}{k^*} = (1-\alpha).
\]
Substitute $L_U^*=\bar L-L_S^*$ and $k^*=\frac{\eta}{\delta_k}L_S^*$, then solve the resulting scalar equation for $L_S^*$. The solution is unique and yields Equation (\ref{eq:s_star}). Since parameters imply $0<s^*<1$, we have $0<L_S^*<\bar L$ and $k^*>0$.
\end{proof}

\subsection{Comparative statics and falsifiable implications}

The closed form labor share in steady state Equation (\ref{eq:s_star}) implies the following.

\begin{proposition}[Value of capability and discounting]
\label{prop:gamma_r}
In the baseline steady state, $s^*$ is increasing in $\gamma$ and decreasing in $r$:
\[
\frac{\partial s^*}{\partial \gamma} > 0,
\qquad
\frac{\partial s^*}{\partial r} < 0.
\]
\end{proposition}

\begin{proposition}[No last mile implies positive maintenance share]
\label{prop:delta_positive}
If $\delta_k>0$, then $s^*>0$. Moreover, $s^*$ is increasing in $\delta_k$:
\[
\frac{\partial s^*}{\partial \delta_k} > 0.
\]
\end{proposition}

\paragraph{Interpretation: capability as codified, partially nonrival knowledge (and why this strengthens No Last Mile).}
Proposition \ref{prop:delta_positive} implies a strictly positive long-run maintenance share whenever effective capability depreciates. A complementary interpretation is that the capability stock $k_t$ is a form of codified organizational knowledge: structured human-data work converts tacit judgment into reusable artifacts (rubrics, evaluation harnesses, exception taxonomies, audit protocols) that can be applied repeatedly across many downstream interactions. This makes $k_t$ \emph{partially nonrival} within the firm or platform in the sense of the endogenous growth literature \citep{romer1990endogenous}: once created, the same codified procedure can improve reliability at scale.

Nonrivalry sharpens the conceptual point in two ways. First, it raises the \emph{social} return to investment in $k_t$ because improvements can be reused across a large volume of production, so the incentives to maintain capability do not vanish as adoption spreads. Second, it highlights a standard wedge: even if a platform internalizes some benefits of $k_t$, it need not internalize all spillovers (across task families, across firms using the same foundation model stack, or through workforce-level incentive misalignment), implying that the competitive/private allocation can \emph{underinvest} relative to the social optimum. Read this way, the baseline $s^\ast$ is best interpreted as a conservative equilibrium lower bound: nonrival scaling forces push toward \emph{more} sustained investment in structured human-data work, not less, and the No Last Mile conclusion survives even when the mechanism is framed as knowledge accumulation rather than “depreciation by assumption.”

\paragraph{Empirical content.}
Propositions \ref{prop:gamma_r}--\ref{prop:delta_positive} are falsifiable once one specifies empirical proxies for (i) the marginal productivity of capability ($\gamma$), (ii) effective depreciation/obsolescence ($\delta_k$), and (iii) the structured labor share $s$. The model predicts that environments with larger measured productivity gains per unit of structured human-data work (higher effective $\gamma$) and environments with faster drift (higher effective $\delta_k$) sustain a larger long-run allocation of labor to structured human-data production.

% ===========================================
% Section 3: Task Families and Reallocation
% ===========================================
\section{Full Model with Task Families, Task Automation, and Labor Reallocation}
\label{sec:model3}

Section \ref{sec:modelA} aggregates all structured work into a single stock $k_t$. This section provides the microfoundation for $k_t$ using a portfolio of task families and derives reallocation results that capture saturation, frontier expansion, and persistent maintenance demand.

\subsection{Task families and maturity stocks}

Let $\mathcal{J}_t$ denote the set of task families active at time $t$. Each task family represents a bundle of work activities available in the economy. These activities may be unstructured and offline (hence not directly observable in digital traces) or codified into explicit procedures and interfaces. We allow task families to vary in their codification intensity and in their exposure to AI (i.e., the extent to which the work can be automated, assisted, or monitored by AI systems).

Each task family $j\in\mathcal{J}_t$ has a codification (maturity) stock $k_{j,t}\ge 0$ evolving as
\begin{equation}
k_{j,t+1} = (1-\delta_j)\,k_{j,t} + g_j(\ell_{j,t}),
\qquad \delta\in(0,1),
\label{eq:kj_law}
\end{equation}
where $\ell_{j,t}\ge 0$ is structured labor allocated to family $j$ at time $t$ and
\begin{equation}
g:\mathbb{R}_+\to\mathbb{R}_+,\quad
g(0)=0,\quad g'(\cdot)>0,\quad g''(\cdot)<0.
\label{eq:g_assump}
\end{equation}
The concavity of $g$ formalizes diminishing returns to codification effort within a family.

Total structured labor satisfies
\begin{equation}
\sum_{j\in\mathcal{J}_t} \ell_{j,t} = L_{S,t}.
\label{eq:LS_sum}
\end{equation}

\subsection{Aggregate capability}

Aggregate capability $k_t$ in Equation (\ref{eq:production}) is an aggregator of the family-level stocks $\{k_{j,t}\}$.
We assume either addictive aggregation for simplicity or CES aggregation to capture task complementarity (more realistic).

\paragraph{Additive aggregator.}
\begin{equation}
k_t = \sum_{j\in\mathcal{J}_t} \omega_j\,k_{j,t},
\qquad \omega_j>0.
\label{eq:k_additive}
\end{equation}

\paragraph{CES aggregator}
\begin{equation}
k_t =
\left(\sum_{j\in\mathcal{J}_t} \omega_j\,k_{j,t}^{\rho}\right)^{1/\rho},
\qquad \omega_j>0,\ \rho\le 1,\ \rho\neq 0.
\label{eq:k_ces}
\end{equation}
Values $\rho<1$ allow complementarity across task families, capturing bottlenecks where weak links have outsized effects.

\subsection{Frontier expansion (entry of new task families)}

Innovation and changing environments introduce new task families. We model this by allowing $\mathcal{J}_t$ to expand:
\begin{equation}
\mathcal{J}_{t+1} = \mathcal{J}_t \cup \mathcal{N}_{t+1},
\label{eq:J_entry}
\end{equation}
where $\mathcal{N}_{t+1}$ is the set of new families arriving between $t$ and $t+1$. New families start with low maturity (for example $k_{j,t}\approx 0$ for $j\in\mathcal{N}_{t}$). The expected size of $\mathcal{N}_{t}$ is governed by an entry intensity $\mu_t\ge 0$.

\paragraph{Why entry intensity varies across environments.}
The entry intensity $\mu_t$ should be interpreted as the rate at which new AI-exposed task families become economically worth codifying, and it is naturally state dependent. 
\begin{itemize}
    \item First, supply-side capability shocks—such as releases of stronger foundation models, improved toolchains, or new modalities—lower the fixed costs of bringing new families on-platform, raising $\mu_t$ around major release windows.
    \item Second, the same shock induces more entry in settings with higher workflow digitization and modularity, because more tasks are observable, contractible, and integrable; hence digitally intensive sectors should exhibit higher baseline $\mu$ and larger responses to model improvements.
    \item Third, institutional and governance shocks (new policies, safety requirements, or compliance standards) can create entire new families of evaluation and auditing tasks, again spiking $\mu_t$.
\end{itemize}
  These mechanisms map to observable variation: task-family “births” on platforms can be related to dated model-release events, sectoral digital-intensity measures, and regulatory or policy shocks.

\subsection{Within-period allocation and reallocation}

To characterize reallocation, consider a reduced-form objective that chooses $\{\ell_{j,t}\}$ to maximize the incremental value of capability accumulation. Let $\Lambda_t>0$ denote the shadow value of aggregate capability at time $t$ (which in the baseline is linked to $\partial Y_t/\partial k_t$). Under the additive aggregator Equation (\ref{eq:k_additive}), maximizing the period-$t$ contribution to aggregate capability is equivalent to
\begin{equation}
\max_{\{\ell_{j,t}\ge 0\}}
\sum_{j\in\mathcal{J}_t} \Lambda_t\,\omega_j\, g(\ell_{j,t})
\quad \text{s.t.} \quad \sum_{j\in\mathcal{J}_t} \ell_{j,t}=L_{S,t}.
\label{eq:allocation_problem}
\end{equation}

\begin{lemma}[Marginal-equalization across task families]
\label{lem:allocation_rule}
Assume Equation (\ref{eq:g_assump}). In any interior optimum of Equation (\ref{eq:allocation_problem}), for all task families $i,j$ with $\ell_{i,t}>0$ and $\ell_{j,t}>0$,
\begin{equation}
\omega_i\, g'(\ell_{i,t}) = \omega_j\, g'(\ell_{j,t}).
\label{eq:marginal_equalization}
\end{equation}
\end{lemma}

\begin{proof}
Form the Lagrangian for Equation (\ref{eq:allocation_problem}). The first-order conditions imply
$\Lambda_t\,\omega_j\,g'(\ell_{j,t})=\nu_t$ for all $j$ with $\ell_{j,t}>0$, where $\nu_t$ is the multiplier on Equation (\ref{eq:LS_sum}). Cancel $\Lambda_t$ to obtain Equation (\ref{eq:marginal_equalization}).
\end{proof}

Lemma \ref{lem:allocation_rule} implies systematic reallocation toward families with higher weights $\omega_j$ and toward families where marginal codification returns remain high, which in turn depends on maturity and how $\omega_j$ and effective drift load onto the aggregator.

\subsection{Roy-style Wage dispersion across task families} \label{subsection:roy}

To microfound dispersion within structured work, let workers $i$ differ in comparative advantage across task families, summarized by productivity shifters $\{a_{ij}\}_{j\in J_t}$ (domain knowledge, calibrated judgment, writing skill, etc.). Suppose one unit of worker time in family $j$ produces “effective structured effort” $e_{ij,t}=a_{ij}$ that is valued at a family-specific piece rate $p_{j,t}$ determined by the marginal value of improving maturity in that family. Under the additive capability aggregator, a natural object is the family shadow price
\[
p_{j,t} \equiv \Lambda_t \omega_j g'(\ell_{j,t}),
\]
so that an individual’s wage if assigned to family $j$ is
\[
w_{ij,t}=p_{j,t}a_{ij}.
\]
In a Roy-style assignment equilibrium, workers select into the family that maximizes $w_{ij,t}$, implying systematic wage differences across families because both (i) prices $p_{j,t}$ vary with bottlenecks, maturity, and frontier status, and (ii) the distribution of $a_{ij}$ differs across families. High-judgment families command premia because they carry high marginal value (high $\omega_j$, low $k_{j,t}$, or complementarity-driven bottlenecks under CES aggregation) while drawing from a thin right tail of $a_{ij}$. By contrast, standardized families face thick global supply and are priced closer to the outside option for general labor, producing a compressed wage distribution. 

This sorting mechanism connects the labor-economics notion of comparative advantage directly to the macro object $L_{S,t}$: the model predicts both a positive structured labor share in the aggregate and a right-skewed within-structured wage distribution driven by complex, frontier and bottleneck families.

\subsection{Falsifiable propositions from task families}

\begin{proposition}[Reallocation toward low-maturity/frontier families]
\label{prop:reallocation_frontier}
Suppose (i) $g$ is concave and (ii) new families arrive with low maturity and high marginal codification value. Then following an increase in frontier entry intensity $\mu_t$ (larger $\mathcal{N}_{t+1}$), the optimal allocation shifts structured labor toward newly arrived and low-maturity task families, holding $L_{S,t}$ fixed.
\end{proposition}

\begin{proposition}[Saturation reduces labor required per incremental improvement]
\label{prop:saturation}
Within any task family $j$, diminishing returns $g''<0$ imply that the marginal capability gain per unit of structured labor, $\partial k_{j,t+1}/\partial \ell_{j,t}=g'(\ell_{j,t})$, declines as $\ell_{j,t}$ rises. Consequently, as a family becomes mature (high $k_{j,t}$ and stable standards), maintaining it requires less incremental labor per unit of capability improvement, and structured effort reallocates to other families with higher marginal gains.
\end{proposition}

\begin{proposition}[Maintenance demand with drift]
\label{prop:maintenance}
If $\delta>0$, then maintaining a stationary distribution of maturity across task families requires a strictly positive long-run flow of structured labor allocated to existing families, even absent entry ($\mu_t=0$). With sustained entry ($\mu_t>0$), the long-run structured labor requirement is bounded away from zero by both maintenance and frontier investment needs.
\end{proposition}

\begin{proposition}[Bottlenecks and persistent premia under complementarity]
\label{prop:bottlenecks}
Under the CES aggregator \ref{eq:k_ces} with $\rho<1$, task families act as complements. In this case, families with low maturity can become bottlenecks with high shadow value, implying persistent incentives to allocate structured labor to scarce frontier families even as other families saturate.
\end{proposition}

\begin{proposition}[Within-structured wage dispersion]
\label{prop:wagedispersion}Under the Roy-style assignment described in Section \ref{subsection:roy} (wage dispersion mechanism), equilibrium wages within structured work are dispersed because workers self-select into task families based on comparative advantage. In particular, families with high shadow value (high $\Lambda_t\omega_j$ and/or bottleneck status under complementarity) pay higher equilibrium rates, and dispersion is amplified when high-judgment families draw on a thin right tail of suitable labor.
\end{proposition}

\begin{proposition}[Wage dispersion responds to frontier entry and drift]
\label{prop:dispersionresponse} Following an increase in frontier entry intensity $\mu_t$ (larger $N_{t+1}$) or an increase in effective drift (higher $\delta$), the cross-family distribution of shadow values steepens, implying increased within-structured wage dispersion and an expansion of the right tail of rates concentrated in complex, frontier and bottleneck families.
\end{proposition}

\paragraph{Empirical content.}
Propositions \ref{prop:reallocation_frontier}--\ref{prop:dispersionresponse} are falsifiable with platform- or taxonomy-level data:
(i) hours by task family over time ($\ell_{j,t}$),
(ii) maturity proxies (rubric stability, inter-rater agreement, rework rates, novelty/edge-case flags),
(iii) frontier shocks (new task-family births, major policy or product surface changes),
and (iv) evidence of complementarity/bottlenecks (disproportionate impact of low-maturity families on downstream capability metrics).

\subsection{Link back to Section \ref{sec:modelA}}

Model \ref{sec:model3} provides a microfoundation for the aggregate capability stock $k_t$ in Section \ref{sec:modelA}. The key additional implication is that even when a given task family saturates, aggregate structured labor demand persists because (i) drift requires maintenance and (ii) innovation introduces new families at low maturity. This portfolio logic strengthens the interpretation of a positive long-run structured labor share as a structural convergence outcome rather than a transient phase.

\section{Measurements, Calibration and Empirics}
\subsection{Data and measurement}

This section maps the model objects to observable quantities using datasets that are standard in the AI labor market literature. No single dataset identifies every object in the model; instead, we organize measurement around complementary data layers that can be combined through common occupational (SOC) or sectoral crosswalks.

\subsubsection{Data layers used in the AI labor market literature}

\paragraph{Layer A: occupation--task maps (benchmarking exposure and task content).}
A large literature measures exposure to AI using task/ability descriptions from O*NET mapped to occupations. Examples include (i) AI occupational exposure measures linking O*NET abilities to AI capabilities, (ii) similarity between tech patents and job tasks, and (iii) LLM-specific exposure rubrics that rate which tasks are plausibly affected by LLM systems. These datasets are useful for constructing task exposure, locating where structured work is likely to arise, and linking outcomes to SOC-based wage and employment series.

\paragraph{Layer B: job postings and resume/career-graph data (labor demand, skills, and wages).}
Job postings microdata allow measurement of adoption-related shifts in hiring, the emergence of new roles/titles, changes in required skills, and human data worker compensation at the task level. These data are central for measuring task-family entry in labor demand and rate premia for specialized structured work.

\paragraph{Layer C: online labor markets (high-frequency task prices and dispersion).}
Our platform data provide high-frequency panels of contract types, categories, hourly rates, and realized matching outcomes. These data are particularly well suited for measuring within-structured sector wage dispersion and its response to AI capability shocks.

\paragraph{Layer D: business adoption surveys (diffusion and cross-sector variation).}
National surveys of firms that include AI use questions provide sectoral and size gradients in adoption and can be used to anchor calibration or heterogeneity in $\mu_t$ across environments.

\paragraph{Layer E: within-firm telemetry or field deployments (linking capability to productivity).}
A smaller set of studies use within-firm deployments of generative AI tools to measure productivity and heterogeneity. These environments are ideal for linking capability proxies to output and learning and for validating mechanisms such as ``rare case'' gains and experience gradients.

\subsubsection{Constructing model objects from observables}

We describe primary and secondary proxies for each object.

\paragraph{Structured effort and its long-run share $s_t$.}
\begin{itemize}
    \item \emph{Primary proxy (postings-based):} construct a postings taxonomy of ``structured human-data'' roles using keywords and embeddings that capture unique work activities to improve model capability. Let $\widehat{L}_{S,t}$ be total labor demand (postings or hires) for these categories and $\widehat{L}_t$ total labor demand; define $\hat{s}_t=\widehat{L}_{S,t}/\widehat{L}_t$.
    \item \emph{Alternative proxy (survey-based):} measure the share of worker hours in human-data online gigs relative to the total work hours.
\end{itemize}

\paragraph{Task families $J_t$, entry $\mu_t$, and frontier expansion.}
\begin{itemize}
    \item \emph{Primary proxy (task-family births in demand):} define a task-family as a stable cluster of postings/task descriptions in the structured taxonomy; count new clusters appearing in period $t$ to obtain $\hat{\mu}_t$.
    \item \emph{Secondary proxies:} (i) new job titles/skills that indicate new evaluation or integration families; (ii) new online labor categories or sharp increases in demand for niche categories.
    \item \emph{Validation and classification:} classify entry events by trigger type: capability shocks (major model/tool releases), integration shocks (new product surfaces, digitization), or governance shocks (policy/regulatory changes). This classification makes Proposition \ref{prop:reallocation_frontier} empirically meaningful.
\end{itemize}

\paragraph{Capability stock $k_t$ (latent) and maturity at the task-family level.}
We interpret $k_t$ as an aggregation of family-specific maturity stocks $\{k_{j,t}\}_{j\in J_t}$.
\begin{itemize}
    \item \emph{Primary proxy (family maturity index):} for each family $j$, construct $\hat{k}_{j,t}$ from (i) rubric stability (edit frequency, versioning intensity), (ii) inter-rater agreement or audit disagreement, (iii) rework/revision intensity, and (iv) failure rates on a fixed audit set. Aggregate to $\hat{k}_t$ using weights proportional to family importance, e.g., $\hat{k}_t=\sum_j \omega_j \hat{k}_{j,t}$.
    \item \emph{External validation:} relate $\hat{k}_t$ to observed productivity or quality metrics in within-firm deployments (Layer E) when available.
\end{itemize}

\paragraph{Depreciation/drift $\delta_k$ and its decomposition.}
we treat $\delta_{k}$ as the probability per period (e.g., per month/quarter) that a task family’s capability/maturity drops noticeably. For instance, when a rubric or eval process stops working due to shifting context and needs refresh.
\begin{itemize}
    \item \emph{Primary proxy:} estimate the probability that $\hat{k}_{j,t}$ crosses a degradation threshold (e.g., a sustained increase in failure/rework or drop in agreement). Aggregate these across task families to obtain $\hat{\delta}_t$.
    \item \emph{Decomposition:} assign shocks to (i) technological drift windows (model/tool releases), (ii) environmental drift windows (demand/content regime shifts), and (iii) organizational drift windows (policy/standards changes), yielding $\hat{\delta}_t=\hat{\delta}^{tech}_t+\hat{\delta}^{env}_t+\hat{\delta}^{org}_t$.
    \item \emph{Secondary proxy (postings/rubric churn):} measure the rate of change in evaluation-related skill requirements and compliance language in postings as an external indicator of organizational drift.
\end{itemize}

\paragraph{Family prices $p_{j,t}$ and wage dispersion.}
\begin{itemize}
    \item \emph{Primary proxy (rates by family):} use posted hourly rates/budgets or posted compensation bands (job postings) to estimate the cross-family distribution of rates $\{w_{j,t}\}$ and its dispersion (variance, interquantile ranges, right-tail shares).
    \item \emph{Sorting tests:} relate $w_{j,t}$ to thinness-of-supply proxies (screen pass rates, time-to-staff, credential requirements) and bottleneck proxies (marginal impact on downstream quality metrics). The Roy mechanism predicts both higher premia and stronger rationing signals in frontier and bottleneck families.
\end{itemize}

\paragraph{Outcome mapping (macro and micro).}
At the occupation/sector level, link exposure indices (Layer A) and adoption measures (Layer D) to employment, wage growth, and task reallocation using standard SOC-sector panels. At the within-firm level, link $\hat{k}_t$ (or $\hat{k}_{j,t}$) to productivity and quality outcomes where telemetry exists (Layer E).

\subsection{Calibration: baseline steady-state labor share}
Equation (\ref{eq:s_star}) delivers a closed-form steady-state structured labor share,
\begin{equation*}
s^*=\frac{\gamma \,\delta_k}{\gamma \,\delta_k+(1-\alpha)(r+\delta_k)},
\end{equation*}
where $\alpha$ is capital’s share in production, $r$ is the discount rate, $\delta_k$ is effective depreciation/obsolescence of the capability stock, and $\gamma$ is the elasticity of productivity with respect to capability in $A(k)=\bar A k^\gamma$. 

To illustrate magnitudes, interpret one model period as a year and use standard macro benchmarks $\alpha \in [0.33,0.40]$ and $r \in [0.03,0.05]$. For the two parameters specific to the human-data mechanism, we consider (i) a range of effective drift/obsolescence $\delta_k \in [0.08,0.25]$ (slow to medium drift) consistent with frequent distribution shifts and changing objectives, and (ii) an elasticity $\gamma \in [0.02,0.08]$ spanning a small to modest leverage of capability into productivity. 

We run a Monte Carlo exercise drawing parameters independently from uniform distributions over the above ranges. For each draw we compute the steady-state structured labor share implied by Eq.(\ref{eq:s_star}). A Monte Carlo simulation (200{,}000 draws) implies that
$s$ is concentrated around $\approx 6\%$: the mean is $5.85\%$ and the median is $5.85\%$
(with standard deviation $1.98$ percentage points). The implied central intervals are
$[3.13\%,\,8.53\%]$ for the 10th--90th percentiles and $[2.58\%,\,9.30\%]$ for the 2.5th--97.5th
percentiles. In the simulated draws, $s$ ranges from approximately $1.90\%$ to $10.51\%$.
Equivalently, $\Pr(s>5\%)\approx 62.7\%$ and $\Pr(s>8\%)\approx 17.1\%$ under these priors.

We treat the human-data parameters $(\gamma,\delta_k)$ as deliberately conservative: we impose a low elasticity of productivity with respect to capability, $\gamma\in[0.02,0.08]$, and a slow-to-moderate effective depreciation (drift) rate, $\delta_k\in[0.08,0.25]$. Both restrictions bias the implied structured labor share $s$ downward, so the resulting Monte Carlo distribution can be interpreted as a lower-bound envelope under muted leverage and muted drift.

\paragraph{Interpreting $\gamma$ and why the benchmark range is conservative.}
In the baseline, $\gamma$ is an aggregate elasticity: it compresses the many ways improved capability (higher reliability, broader coverage, fewer exceptions and rework loops) translates into economy-wide TFP on economically meaningful tasks. Because generative AI and adoption is organizationally mediated, macro-oriented assessments can imply small near-term aggregate effects even when micro studies find large gains on narrow task bundles. Choosing $\gamma\in[0.02,0.08]$ is therefore deliberately conservative: even large proportional increases in $k$ translate into modest proportional increases in $A(k)=\bar A k^\gamma$. At the same time, mechanism-consistent evidence suggests that structured human feedback and evaluation effort can move model behavior materially without changing architectures \citep{ouyang2022rlhf}, which is precisely the “codified capability” margin captured by $k_t$. If $\gamma$ were instead moderately larger (e.g., extending the upper bound to 0.12), the same Monte Carlo priors would shift the mean implied structured labor share from about 6\% toward about 8\%, reinforcing that our benchmark should be read as a lower-bound envelope rather than an aggressive forecast.

\subsection{Empirical roadmap (outline for follow-up work)}
This paper is primarily theoretical, but the mechanism yields a concrete empirical agenda. The goal of this section is to make the objects in the model operational and to outline designs that can be executed with platform telemetry and dated product/model events.

\subsubsection{Measurement construction}
(i) Task-family maturity $k_{j,t}$: rubric stability, inter-rater agreement, rework/revision intensity, failure rates on a fixed audit set, and time-to-resolution for escalations. 

(ii) Drift/depreciation $\delta_k$: estimate maturity-loss hazards at the task-family level and aggregate using the capability weights $\omega_j$; decompose into technological, environmental, and organizational drift using event windows around releases, demand/content shifts, and policy/standards changes.

(iii) Frontier entry $\mu_t$: count new task-family ``births'' per period and classify births by trigger type (new model/tool capability, sectoral digitization/integration, governance/policy).

(iv) Shadow value proxies $\Lambda_t\omega_j$: downstream impact metrics (e.g., reduction in rework loops, improved pass rates on audits, increased coverage of previously unhandled cases) and bottleneck indicators under complementarity.

\subsubsection{Identification sketches}
Design A (event studies): estimate responses of $\ell_{j,t}$, $k_{j,t}$, and wage/rate premia to dated shocks: major model/tool releases (technological drift), policy changes (organizational drift), and content-regime shifts (environmental drift). 

Design B (cross-environment comparisons): relate long-run structured labor shares and entry intensity to sectoral digitization and integration readiness, testing whether digitally intensive environments exhibit higher baseline $\mu$ and larger responses to capability shocks.

Design C (sorting tests): test Roy-style predictions by relating family-level rate premia to (a) bottleneck/marginal-value proxies and (b) thinness-of-supply proxies (screen pass rates, staffing delays), and by checking whether rate dispersion widens in periods of high $\mu_t$ or high drift.

\subsubsection{Stylized facts to report in follow-up work}
A small set of descriptive facts would already discipline the mechanism: (1) a time series of task-family births (proxying $\mu_t$) and their clustering around releases/policy changes; (2) maturity-loss hazards for a stable panel of families (proxying $\delta$); (3) within-structured rate dispersion and its concentration in frontier/bottleneck families; and (4) reallocation of structured hours toward newly born or high-drift families following shocks.

\section {Policy and Governance Implications beyond the Last Mile }
The "no last mile" result has direct implications for how policymakers and workforce planners think about AI-driven labor market transitions. The standard framing treats structured human-data work as transitional, a bridge between today's imperfect models and a future state where automation is complete. This paper's framework challenges framing at its foundation. If capability stocks depreciate continuously due to technological, environmental, and organizational drift, then structured human-data work is not a temporary buffer but a permanent structural feature of AI-augmented economies. Labor market policy built on the assumption that this category of work will eventually automate itself risks misallocating training investments, undervaluing the workers performing it, and failing to build the institutional infrastructure needed to sustain quality over time. Workforce planning frameworks should instead treat structured human-data roles as a distinct occupational category requiring career pathways, credentialing, and compensation structures commensurate with their long-run economic function.

At the platform and regulatory level, the "no last mile" result raises questions about how AI systems are governed over their operational lifespan. Current regulatory frameworks tend to focus on model deployment as the critical moment of accountability, with less attention to the ongoing maintenance work that keeps deployed systems aligned with evolving standards, user needs, and legal requirements. If structured human-data work is the mechanism through which alignment is maintained rather than a one-time input at training, then governance frameworks need to account for it explicitly, including transparency requirements around maintenance labor, standards for rubric and evaluation design, and accountability mechanisms when capability stocks are allowed to depreciate through underinvestment. This has implications not only for AI platform operators but for procurement decisions by governments and large institutions that depend on AI systems remaining reliably aligned over multi-year deployment horizons.

\section{Conclusion}
This paper argues that a durable bottleneck in the AI economy is not model architecture, but the ongoing production of \emph{structured human data}: evaluation, rubric-based judgment, auditing, exception handling, and continual updates that convert raw model capability into dependable, deployable performance. We formalize these activities as a distinct input that accumulates a reusable capability stock, raising productivity both by improving reliability on existing tasks and by expanding the frontier of tasks for which AI can be used with confidence. 

Our central theoretical result is the ``no last mile'' implication. Even as foundation models improve, demand for structured human-data work does not disappear in equilibrium because capability depreciates as the world changes: new edge cases emerge, standards evolve, and systems must be continuously monitored and refreshed. In the task-family extension, structured effort reallocates toward low-maturity or newly arriving families, and complementarity can make certain families persistent bottlenecks. The same framework also delivers a simple labor-market prediction: compensation within structured work is naturally dispersed, because workers sort across task families by comparative advantage and high-judgment families combine high marginal value with thin qualified supply. 

A conservative calibration illustrates magnitude: under muted leverage of capability into productivity and slow-to-moderate drift, the implied long-run structured labor share is centered around mid-single digits, with a meaningful right tail. Interpreted as a lower-bound envelope, these magnitudes are consistent with a world in which near-term aggregate labor-market effects can look small while the underlying reorganization margin is active and persistent.

Several limitations are important. The model is intentionally stylized: we compress heterogeneous deployment settings into a small set of parameters; we treat the mapping from capability to productivity in reduced-form; and we abstract from platform market power, contracting frictions, and general-equilibrium feedbacks that could amplify or dampen structured labor demand. On the measurement side, the key objects (capability maturity, drift, task-family entry, and shadow values) require combining multiple data layers that are not always jointly observable, and the cleanest validation often depends on within-firm telemetry that is difficult to access at scale.

These limitations point to an agenda rather than a gap. The framework suggests concrete tests using (i) event studies around model/tool releases and policy changes to estimate drift and reallocation, (ii) job-postings and online labor market data to measure task-family births and wage dispersion, and (iii) within-firm deployments to link capability proxies to productivity and learning. More broadly, the ``no last mile'' result reframes both workforce planning and AI governance: structured human-data work is not a temporary bridge to full automation, but a persistent input into sustaining alignment and reliability over the lifecycle of deployed systems.

\bibliography{nolastmile}
\bibliographystyle{apalike}

\end{document}